\documentclass[floatfix, noshowpacs, preprintnumbers, twocolumn, amsmath, amssymb, aps, superscriptaddress, prl]{revtex4-1}

\pdfoutput=1

\usepackage{lmodern}
\usepackage{graphicx, xcolor}
\usepackage[caption=false]{subfig}
\usepackage{tabularx}
\usepackage{soul}

\usepackage[xcolor=dvipdf]{changes}
\definechangesauthor[name={Luca Innocenti}, color={orange}]{LI}
\setremarkmarkup{ {\color{blue}\bfseries\footnotesize (#2)} }
\setaddedmarkup{{\color{green!90!black}#1}}
\setdeletedmarkup{{\color{red}\st{#1}}}

\begin{document}

\title{Experimental generalized quantum suppression law in Sylvester interferometers}

\author{Niko Viggianiello}
\affiliation{Dipartimento di Fisica, Sapienza Universit\`{a} di Roma,
Piazzale Aldo Moro 5, I-00185 Roma, Italy}

\author{Fulvio Flamini}
\affiliation{Dipartimento di Fisica, Sapienza Universit\`{a} di Roma,
Piazzale Aldo Moro 5, I-00185 Roma, Italy}

\author{Luca Innocenti}
\affiliation{Dipartimento di Fisica, Sapienza Universit\`{a} di Roma,
Piazzale Aldo Moro 5, I-00185 Roma, Italy}
\affiliation{Centre for Theoretical Atomic, Molecular and Optical Physics,
School of Mathematics and Physics, Queen's University, Belfast BT7 1NN, United Kingdom}

\author{Daniele Cozzolino}
\affiliation{Dipartimento di Fisica, Sapienza Universit\`{a} di Roma,
Piazzale Aldo Moro 5, I-00185 Roma, Italy}
\affiliation{Dipartimento di Fisica, Universit\`{a} degli Studi di Napoli Federico II,
Corso Umberto I, 40, 80138 Napoli, Italy}

\author{Marco Bentivegna}
\affiliation{Dipartimento di Fisica, Sapienza Universit\`{a} di Roma,
Piazzale Aldo Moro 5, I-00185 Roma, Italy}

\author{Nicol\`o Spagnolo}
\affiliation{Dipartimento di Fisica, Sapienza Universit\`{a} di Roma,
Piazzale Aldo Moro 5, I-00185 Roma, Italy}

\author{Andrea Crespi}
\affiliation{Istituto di Fotonica e Nanotecnologie, Consiglio Nazionale delle Ricerche (IFN-CNR), 
Piazza Leonardo da Vinci, 32, I-20133 Milano, Italy}
\affiliation{Dipartimento di Fisica, Politecnico di Milano, Piazza Leonardo da Vinci, 32, I-20133 Milano, Italy}

\author{Daniel J. Brod}
\affiliation{Perimeter Institute for Theoretical Physics, 31 Caroline Street North, Waterloo, ON N2L 2Y5, Canada}
\affiliation{Instituto de F\'isica, Universidade Federal Fluminense, 
Av. Gal. Milton Tavares de Souza s/n, Niter\'oi, RJ, 24210-340, Brazil}

\author{Ernesto F. Galv\~{a}o}
\affiliation{Instituto de F\'isica, Universidade Federal Fluminense, 
Av. Gal. Milton Tavares de Souza s/n, Niter\'oi, RJ, 24210-340, Brazil}

\author{Roberto Osellame}
\affiliation{Istituto di Fotonica e Nanotecnologie, Consiglio
Nazionale delle Ricerche (IFN-CNR), Piazza Leonardo da Vinci, 32,
I-20133 Milano, Italy}
\affiliation{Dipartimento di Fisica, Politecnico di Milano, Piazza
Leonardo da Vinci, 32, I-20133 Milano, Italy}

\author{Fabio Sciarrino}
\affiliation{Dipartimento di Fisica, Sapienza Universit\`{a} di Roma,
Piazzale Aldo Moro 5, I-00185 Roma, Italy}

\begin{abstract}

Photonic interference is a key quantum resource for optical quantum computation, and in particular for so-called boson sampling machines. In interferometers with certain symmetries, genuine multiphoton quantum interference effectively suppresses certain sets of events, as in the original Hong-Ou-Mandel effect. Recently, it was shown that some classical and semi-classical models could be ruled out by identifying such suppressions in Fourier interferometers. Here we propose a suppression law suitable for random-input experiments in multimode Sylvester interferometers, and verify it experimentally using 4- and 8-mode integrated interferometers. The observed suppression is stronger than what is observed in Fourier interferometers of the same size, and could be relevant to certification of boson sampling machines and other experiments relying on bosonic interference. 
\end{abstract}

\maketitle

\textit{Introduction. ---} Scalable, general-purpose quantum computers, once developed, will be able to solve problems thought to be intractable for ordinary computers. Given the significant technological challenges involved, other nearer-term goals for the field were suggested, such as the creation of quantum machines able to beat classical computers in particular computational tasks. One such proposal which has drawn much interest is boson sampling \cite{Aaronson10}, which relies on multiphoton interference in a random linear interferometer, and whose output statistics are thought to be hard to sample from classically. Due to bosonic interference, the output of those experiments is distributed according to the permanents of complex matrices specifying the interferometer's design, and the permanent is a function that is notoriously hard to calculate \cite{Valiant79, Troyansky96, Aaronson10}. This possible path towards a demonstration of quantum computational supremacy has resulted in first experimental implementations of such boson sampling computers \cite{Broome13, Spring13, Tillmann12, Crespi12, Spagnolo13birthday, Spagnolo14, Carolan14, Carolan15, Bentivegna15,Loredo2017,Wang2017}. 

\begin{figure*}[ht!]
\centering
\includegraphics[width=0.99\textwidth]{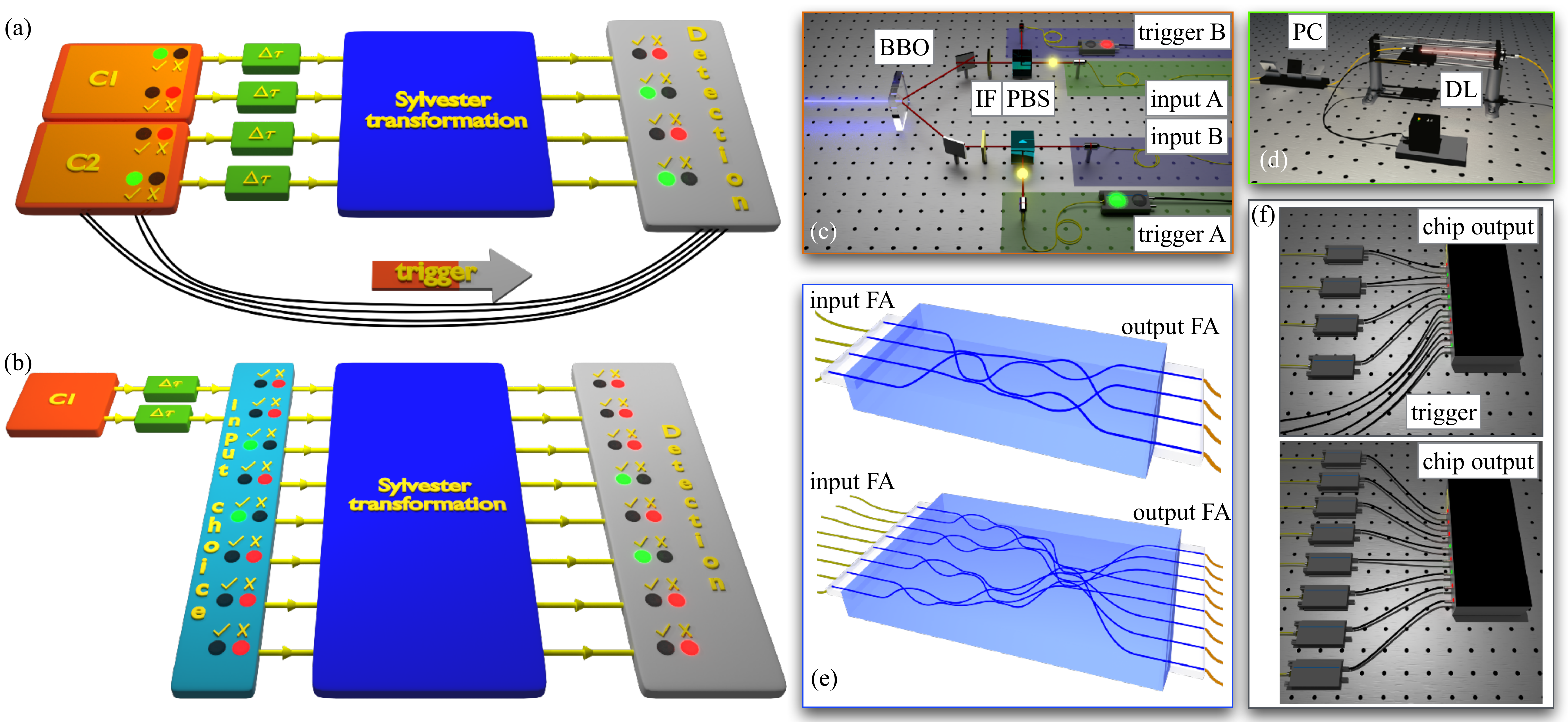}
\caption{(a) Scheme for the 4-mode and (b) for the 8-mode experiments. (c) Generation of two-photon states with four type-II parametric downconversion (PDC) sources embedded in crystals C1 e C2 (a) for the scattershot regime and for the single-source case with Crystal C1 (b). By exploiting the polarization degree of freedom, each crystal comprises two different two-photon sources. In the scattershot regime, the injection of single photons in specific input modes of the interferometer is heralded by the detection of the corresponding twin photon. In the single-source case (b), a pair of twin photons are directly injected into the interferometers without heralding detectors, while the input states are manually changed and characterized separately. (d) Injected photons are synchronized with a specific delay $\Delta t$ so as to be effectively indistinguishable. (e) Layout of the 4- and 8-mode Sylvester interferometers. Photon injection and extraction is done using FAs. (f) Each output is collected by an avalanche photodiode, while electric signals are sent to an electronic stage to analyze data. The acquisition systems are shown for the 4-mode device (top) and 8-mode device (bottom). [BBO: beta-barium borate (source), IF: interferential filter,  PBS: polarizing beam splitter, PC: polarization compensator, DL: delay lines with motorized stages, FA: fiber array]}
\label{setup}
\end{figure*}
 
To investigate the factors behind the computational complexity of these devices, different kinds of input states have been considered. It is known that inputs consisting of distinguishable photons, coherent states, and antisymmetric (fermionic-like) states all result in classically simulable behavior. On the other hand, photon-added coherent states \cite{Sesh2015} and quantum superpositions of coherent states (cat states) \cite{Rohde2015} have been claimed to yield hard-to-simulate outputs. Partially distinguishable photons seem to yield an intermediate regime deserving further investigation \cite{Tichy2015,Till2015,Shch2015}. The effect of losses on simulation complexity was also considered \cite{Aaronson15,RahimiKeshari16}. Moreover, specific semi-classical states able to reproduce some collective interference effects while being efficiently computable have been identified \cite{Tichy10}. 

Other implementations of boson sampling devices serve to improve performance, or were used to discuss how to certify the correct functioning of the device. So-called scattershot boson sampling devices \cite{AaronsonBlog,Lund14,Bentivegna15} use simultaneous pumping of several parametric down-conversion sources to result in a random-input version of the original problem. The main advantage is a greatly enhanced generation rate, compared to a single-input implementations using the same sources. While complete certification of large boson sampling devices may turn out to be impossible, a number of proposals for partial certification were made, capable of comparing the experimental outcomes against some physically motivated error models \cite{Gogolin13, Aaronson14, Spagnolo14, Carolan14, Bentivegna14, Walschaers16, Bentivegna2016, Wang2016, Shch2016}.
 
As genuine many-particle interference is required for boson sampling, efforts have been directed at providing stronger evidence of this phenomenon. The simplest such demonstration is the Hong-Ou-Mandel (HOM) effect \cite{Hong87}. Generalisations of this effect have been proposed, to investigate signatures of interference for specific input-output combinations of Fourier \cite{Tichy10,Tichy14,Crespi16} and Sylvester \cite{Crespi15} interferometers. These highly-symmetric transformations provide a rich landscape for multi-particle interference to happen, which we investigate here.

In this Letter we discuss and experimentally demonstrate a novel zero-transmission law for Sylvester interferometers. As we will see, indistinguishable photons interfere in these devices so that a certain fraction of all input-output combinations are suppressed.  This suggests they may be helpful in identifying multi-photon interference in random-input, scattershot boson sampling experiments. The new zero-transmission law is demonstrated experimentally in Sylvester interferometers with $m=4$ and $m=8$ modes, implemented with the 3-dimensional capabilities of femtosecond laser writing technology. 

\begin{figure*}[ht!]
	\includegraphics[width=0.99\textwidth]{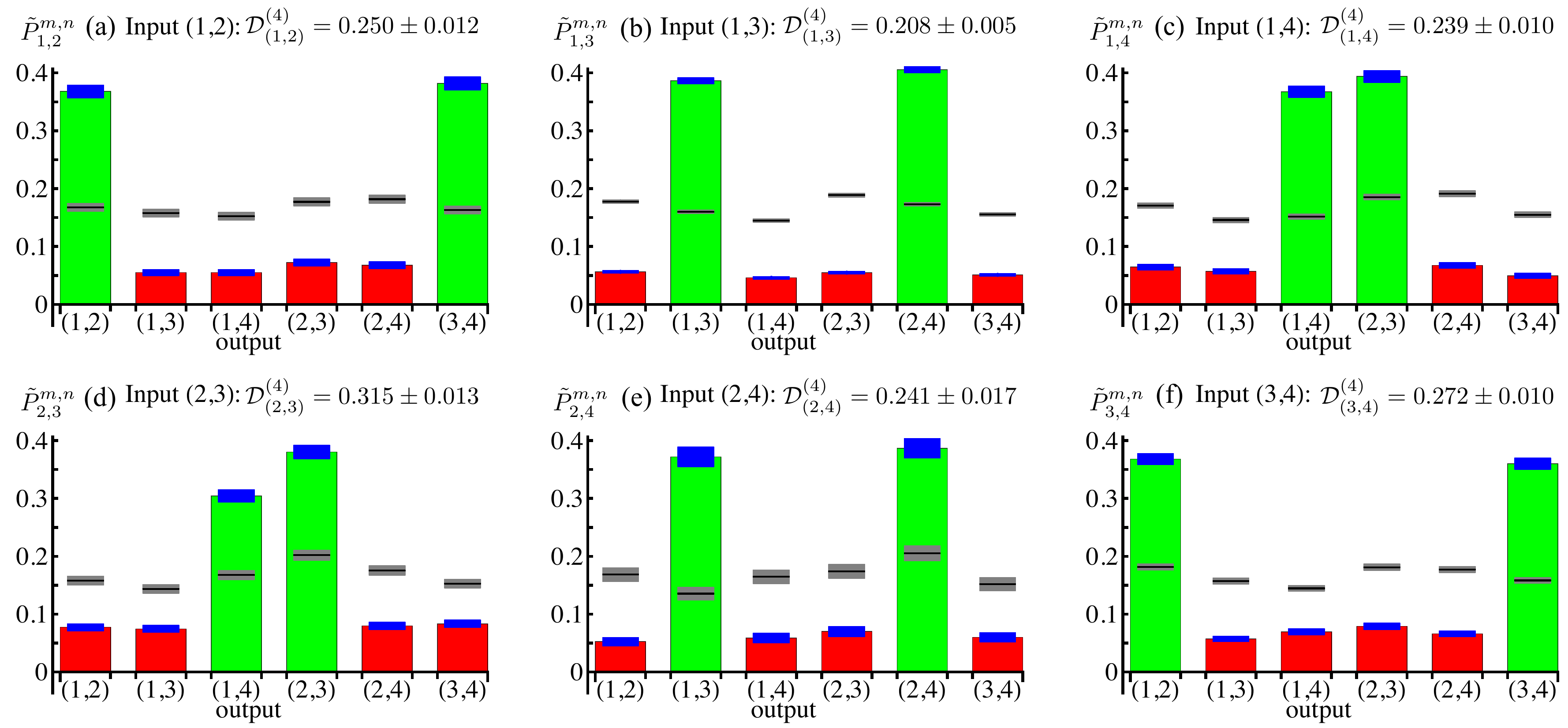}
	\caption{Measured two-photon probability distributions $\tilde{P}^{m,n}_{i,j}$ in the scattershot regime for all input-output combinations with the 4-mode device, where $(i,j)$ denote the input modes and $(m,n)$ the output modes. Input-output combinations forbidden by the suppression law are shown in red, while non-suppressed ones are shown in green. Error bars (represented as blue rectangles) are due to Poissonian uncertainties. Black horizontal segments: measured probabilities with distinguishable particles, with corresponding error bars (grey rectangles). In each subfigure the single violations $D^{(4)}_{(i,j)}$ for each input state are reported.}	
\label{statistic}
\end{figure*}

\textit{Suppression law in Sylvester matrices. ---}
Sylvester matrices ($H$) are special constructions of Hadamard matrices of dimension $m=2^p$, with $p$ integer, whose elements are only +1 and -1 and can be derived recursively as:
\begin{equation}
H(2^p)=
\begin{pmatrix}
H(2^{p-1}) & H(2^{p-1}) \\
H(2^{p-1}) & -H(2^{p-1})
\end{pmatrix}
\end{equation}
The associated unitary transformation is a rescaled $U_{\mathrm{S}}^{m} = m^{-1/2} H(m)$. When $p=1$ we have the simplest Sylvester unitary, describing a symmetric 50:50 beam splitter. The celebrated Hong-Ou-Mandel effect  \cite{Hong87} is a first example of suppression law, as a transition between photons in input modes $(1,2)$ and output modes $(1,2)$ is strictly suppressed. In \cite{Crespi15}, more general input states were considered, consisting of $n = 2^q$ photons, described by a list of mode occupation numbers $(1 + nc, \ldots , n + nc)$ with $0 \le c \le 2^k - 1$, injected in a $2^{k+q}$-dimensional Sylvester interferometer. It was proven that all outputs whose bitwise sum of binary-represented mode occupation numbers is not zero are strictly suppressed. This suppression law only identifies a small subset of all suppressed input-output combinations.

We now introduce a generalized law that identifies a larger number of suppressed input-output combinations in Sylvester interferometers. A $n$-photon state $\vert r \rangle$ may be represented by a mode assignment list $\tilde{r}=(\tilde{r}_1, \tilde{r}_2, \ldots, \tilde{r}_n)$, specifying which modes are occupied. It will be useful to represent this state also as a $n\times p$ binary matrix $R$, whose $i$-th row is the binary representation of $\tilde{r}_i -1$. Let us call $\mathcal{N}^A(R)$ the matrix obtained by negating the columns of $R$ specified by a list of indices $A$. The input-output combinations $\vert r \rangle \rightarrow \vert s \rangle$ and $\vert s \rangle \rightarrow \vert r \rangle$ are suppressed if some $A$ exists for which the following two conditions are met \cite{NoteOdd}:
\begin{gather}
\mathcal{N}^A(\tilde{r}) = \tilde{r} \\
\bigoplus_{k = 1}^n \bigoplus_{\alpha \,\in A} S_{k,\alpha} = 1
\end{gather}
where $\oplus$ denotes the bitwise sum, and $S$ is the binary matrix representation of state $\vert s \rangle$. In \cite{SuppMat} we give a proof of this law (in Sec. I B), as well as a calculation of the asymptotic behavior for the fraction of suppressed events these criteria identify (in Sec. I C). Note that similar criteria were independently reported in Ref. \cite{Dittel2017} while this letter was under completion. For the experimental implementations we report on below, our criteria identify all suppressed events, which are more numerous than the equivalent result for Fourier matrices \cite{Tichy14} with the same size (See Sec. I D of \cite{SuppMat} for more details).

\textit{Experimental generalized suppression law. ---} 
The generalized suppression law has been tested experimentally in 4-mode and 8-mode Sylvester interferometers, implemented by exploiting a 3D architecture enabled by femtosecond laser writing \cite{Crespi16} (see Fig. \ref{setup}). We performed a two-photon scattershot boson sampling experiment \cite{Bentivegna15} to verify the suppression law, feeding each input port of the 4-mode chip with one heralded photon from a different PDC pair. The output events corresponding to two-photon injection were then post-selected via four-fold coincidence measurements (two heralding detectors and two detectors at the output of the device) for all the $\binom{4}{2}=6$ input-output combinations. The experimental setup adopted for the scattershot approach is shown in Fig. \ref{setup}. Further details on the generation and detection are reported in Sections II and III of \cite{SuppMat}.

\begin{figure*}[ht]
	\includegraphics[width=0.99\textwidth]{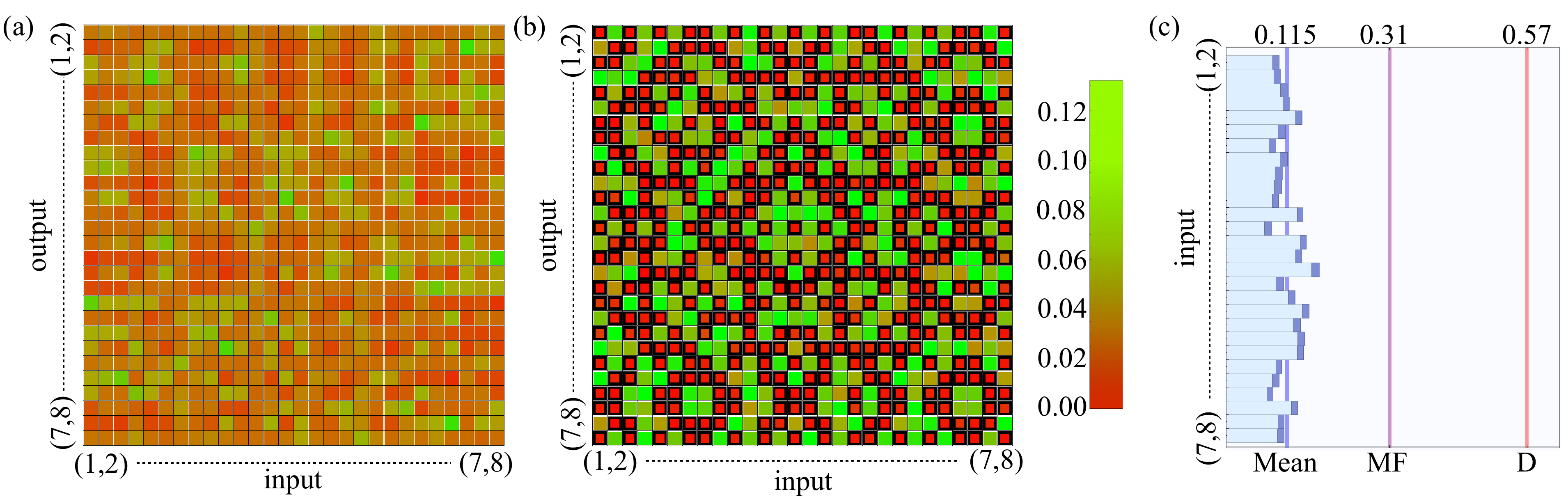}
	\caption{
		(a) Two-photon probability distribution in the regime of distinguishable photons, 8-mode device.
		(b) Two-photon probability distribution in the regime of indistinguishable photons in the 8-mode device, calculated via the HOM visibilities measured for all input-output combinations. Combinations forbidden by the suppression law are shown in black squares, while measured probabilities for each event are shown with a scale from red (lower values) to green (higher values).
		(c) Measured violations for all the input combinations. Cyan bars are the observed violations, with corresponding error bars in blue. Blue line: mean over all violations. Purple line: Mean Field (MF) threshold. Red line: distinguishable photons (D) threshold.}
\label{28-ind-dist-violations}
\end{figure*}

The observed two-photon output statistics are shown in Fig. \ref{statistic}, where one can recognize the pattern of peaks and dips predicted by the Sylvester matrix ($U_{\mathrm{S}}^{4}$). An effective figure of merit to estimate if the observed data are compatible with those expected from a fully interfering multiphoton source is the degree of violation  $\mathcal \nu = N_{\operatorname{forbidden}}/N_{\operatorname{events}}$ \cite{Tichy14,Crespi16}, i.e. the ratio between the number of observed events in suppressed input-output combinations and the total number of events. Evaluating the observed violation permits to rule out alternative hypotheses on the nature of the injected state. The simplest alternative hypothesis is that the $n$ photons are distinguishable. A more elaborate alternative model is the Mean Field (MF) state \cite{Tichy14,Hadzibabic04,Cennini05}, defined as a single-particle state whose wavefunction is macroscopically spread over a set of inputs $A$: $|\psi _{\mathrm{MF}}^A\rangle = n^{-1/2} \sum_{r \in A} e^{i\theta_r}| j_r \rangle,$ where $| j_r\rangle$ identifies a single-photon state in mode $j_r$, and phases $\theta_r$ are randomly chosen from a uniform distribution. Summing the output statistics of $n$ states of this form reproduces some macroscopic interference effects \cite{Carolan14} of $n$ indistinguishable photons injected in modes $A$. Note that suppression laws are only partially fulfilled by MF states statistics, which confirms the diagnostic power of these tools. In particular, it is easy to check that the expected degree of violation for distinguishable particles and MF states in $U_{S}^{4}$ are respectively $0.66$ and $0.4$ (considering only collision-free two-photon inputs and outputs). The overall experimental violation $\mathcal{\nu}$, obtained by summing all the measured events over all possible input-output combinations, is found to be $\mathcal{\nu}^{(4)} = 0.238 \pm 0.003$. The measured value, below the two thresholds by respectively $141$ (distinguishable) and $54$ (MF) standard deviations, thus unambiguously excludes both these hypotheses (Fig. \ref{statistic}). 

The suppression law has been experimentally verified also using a Sylvester 8-mode chip, where the full set of $\binom{8}{2}=28$ two-photon input states were independently investigated. In Fig. \ref{28-ind-dist-violations} we report the full set of $28\times 28$ experimental probabilities, retrieved from the HOM visibilities measured for all no-collision input-output combinations. In this case, one can show that the theoretical degrees of violation for distinguishable particles and MF states are respectively $0.57$ and $0.31$. The mean value of the violations is $\bar{\nu}^{(8)}=0.115\pm 0.002$, well below the theoretical bounds predicted for MF and distinguishable particles, thus ruling out these alternative hypotheses.

\textit{Conclusions and discussion. ---} We introduced and verified experimentally a novel suppression law pertinent to Sylvester interferometers with inputs of indistinguishable photons. The choice of this unitary transformation ensures a significant fraction of suppressed input-output combinations, even for a larger number of photons and modes. For example, for $n=6$ and $m=16$ our law predicts the suppression of $\sim 10\%$ of the output events, compared to a total fraction of forbidden events of $\sim 37\%$ (computed numerically), against just $\sim 2.3 \%$ for the Fourier matrix. This suppression can be used to certify genuine multiphoton interference in scattershot boson sampling experiments, where the full set of possible states is employed. Indeed, random-input generation has been recognized as a promising approach to greatly enhance the event rate in PDC-based scattershot boson sampling experiments. We have shown that the suppression law is able to rule out alternative models which show coarse-grained interference effects, which will be helpful in testing future boson sampling experiments.

\begin{acknowledgments}
\textit{Acknowledgements. ---} This work was supported by the ERC-Starting Grant 3D-QUEST (3D-Quantum
Integrated Optical Simulation; grant agreement no. 307783): http://www.3dquest.eu and by the H2020-FETPROACT-2014 Grant QUCHIP (Quantum Simulation on a Photonic Chip; grant agreement no. 641039): http://www.quchip.eu. D.B. acknowledges Brazilian funding agency CAPES and the Perimeter Institute for Theoretical Physics for financial support. E.F.G. acknowledges support by Brazilian funding agency CNPq.
\end{acknowledgments}


\begin{thebibliography}{40}
\bibitem[{Not()}]{NoteOdd}
\bibinfo{note}{The conditions of the generalized suppression law are never
  satisfied when $n$ is odd.}
\end{thebibliography}
\end{document}


\title{Experimental generalized quantum suppression law in Sylvester interferometers - Supplemental Material}

\author{Niko Viggianiello}
\affiliation{Dipartimento di Fisica, Sapienza Universit\`{a} di Roma,
Piazzale Aldo Moro 5, I-00185 Roma, Italy}

\author{Fulvio Flamini}
\affiliation{Dipartimento di Fisica, Sapienza Universit\`{a} di Roma,
Piazzale Aldo Moro 5, I-00185 Roma, Italy}

\author{Luca Innocenti}
\affiliation{Dipartimento di Fisica, Sapienza Universit\`{a} di Roma,
Piazzale Aldo Moro 5, I-00185 Roma, Italy}
\affiliation{Centre for Theoretical Atomic, Molecular and Optical Physics,
School of Mathematics and Physics, Queen's University, Belfast BT7 1NN, United Kingdom}

\author{Marco Bentivegna}
\affiliation{Dipartimento di Fisica, Sapienza Universit\`{a} di Roma,
Piazzale Aldo Moro 5, I-00185 Roma, Italy}

\author{Nicol\`o Spagnolo}
\affiliation{Dipartimento di Fisica, Sapienza Universit\`{a} di Roma,
Piazzale Aldo Moro 5, I-00185 Roma, Italy}

\author{Andrea Crespi}
\affiliation{Istituto di Fotonica e Nanotecnologie, Consiglio Nazionale delle Ricerche (IFN-CNR), 
Piazza Leonardo da Vinci, 32, I-20133 Milano, Italy}
\affiliation{Dipartimento di Fisica, Politecnico di Milano, Piazza Leonardo da Vinci, 32, I-20133 Milano, Italy}

\author{Daniel J. Brod}
\affiliation{Perimeter Institute for Theoretical Physics, 31 Caroline Street North, Waterloo, ON N2L 2Y5, Canada}
\affiliation{Instituto de F\'isica, Universidade Federal Fluminense, 
Av. Gal. Milton Tavares de Souza s/n, Niter\'oi, RJ, 24210-340, Brazil}

\author{Ernesto F. Galv\~{a}o}
\affiliation{Instituto de F\'isica, Universidade Federal Fluminense, 
Av. Gal. Milton Tavares de Souza s/n, Niter\'oi, RJ, 24210-340, Brazil}

\author{Roberto Osellame}
\affiliation{Istituto di Fotonica e Nanotecnologie, Consiglio
Nazionale delle Ricerche (IFN-CNR), Piazza Leonardo da Vinci, 32,
I-20133 Milano, Italy}
\affiliation{Dipartimento di Fisica, Politecnico di Milano, Piazza
Leonardo da Vinci, 32, I-20133 Milano, Italy}

\author{Fabio Sciarrino}
\affiliation{Dipartimento di Fisica, Sapienza Universit\`{a} di Roma,
Piazzale Aldo Moro 5, I-00185 Roma, Italy}

\maketitle

\section{Suppression laws for Sylvester interferometers}

In this Section we discuss in detail the suppression law for Sylvester interferometers adopted in the main text.

\subsection{Preliminary definitions}
\label{sec:notation}

Let us begin by defining some notation used throughout this Section. We are interested in the transition amplitudes between states of $n$ photons in $m$ modes, and restrict ourselves to the case $n \leq m$. We denote such states as $\ket{r} := \ket{r_1,\dots,r_m}$, where $r_k$ is the number of photons in the $k$th mode and $\sum_{k=1}^m r_k = n$. The set of all such states is denoted by $G_{n,m}$. We refer to the particular states where all $r_i$s take only values 0 or 1 as collision-free, and the set of all collision-free states is denoted by $Q_{n,m}$. It is easy to see that $|G_{n,m}|=\binom{m+n-1}{n}$ and $|Q_{n,m}|=\binom{m}{n}$. Let us now define two useful alternative representations for such states. 
\begin{definition} \label{def:Reps}
(MAL and BM representations). Let $\ket{r} = \ket{r_1,\dots,r_m} \in G_{n,m}$. We can define the following two representations of this state.
\begin{itemize}
\item Mode Assignment List (MAL): an $n$-dimensional vector $\tilde r := (\tilde r_1,\dots, \tilde r_n)$ listing the mode occupied by each particle.
\item Binary matrix (BM): States with $m = 2^p$, for some positive integer $p$, can also be represented as an $n \times p$ binary matrix $R$ as follows.
For a given MAL $\tilde r := (\tilde r_1,\dots, \tilde r_n)$ describing the state, the $i$th row of $R$ corresponds to the binary representation of $\tilde r_i-1$ (padded with zeros on the left so it has length $p$). 
\end{itemize}
Since photons are indistinguishable, the ordering of elements in the MAL is meaningless, and we conventionally choose the list to be in increasing order. Correspondingly, the ordering of the rows of two binary matrices is irrelevant, i.e.\ if $R$ and $R'$ are related only by a permutation of the rows, they represent the same physical state, and this we denote as $R \sim R'$.
\end{definition}

\begin{leftbar}
\textbf{Example 1}. \textit{For $m=8$ and $n=4$, consider the state $\ket{r} = \ket{1,1,1,0,0,0,0,1}$. It has the corresponding MAL $\tilde r = (1,2,3,8)$, and
its BM representation is the $4\!\times\!3$ matrix
\begin{equation}
R = \begin{pmatrix} 0&0&0\\0&0&1\\0&1&0\\1&1&1 \end{pmatrix}.
\end{equation}
}
\end{leftbar}

The following definitions are also convenient shorthands used throughout this Section.

\begin{definition} \label{def:negation}
Let $R$ be the $n\!\times\! p$ BM representation of some state $\ket{r}$. Let $A$ be some subset of the columns of $R$. We denote by $\mathcal{N}^A(R)$ the matrix obtained by flipping each bit in the columns of $R$ specified in $A$. If $\tilde r$ is the MAL associated with $R$, we denote by $\mathcal{N}^A(\tilde r)$ the MAL associated with $\mathcal{N}^A(R).$
\end{definition}

\begin{leftbar}
\textbf{Example 2}. \textit{
For $m=8$ and $n=4$, consider the state $\ket{r}=\ket{1,1,1,1,0,0,0,0}$, corresponding MAL $\tilde r=(1,2,3,4)$ and BM representation
\begin{equation}
R = \begin{pmatrix} 0&0&0\\0&0&1\\0&1&0\\0&1&1 \end{pmatrix}.
\end{equation}
Then we have the following possibilities when $A$ consists of a single element
\begin{equation}
\mathcal N^{\{1\}}(R) = \begin{pmatrix} 1&0&0\\1&0&1\\1&1&0\\1&1&1 \end{pmatrix},
\quad
\mathcal N^{\{2\}}(R) = \begin{pmatrix} 0&1&0\\0&1&1\\0&0&0\\0&0&1 \end{pmatrix},
\quad
\mathcal N^{\{3\}}(R) = \begin{pmatrix} 0&0&1\\0&0&0\\0&1&1\\0&1&0 \end{pmatrix}.
\end{equation}
Clearly $\mathcal N^{\{2\}}(R) \sim \mathcal N^{\{3\}}(R) \sim R$, since they are related to each other by permutations of their rows, whereas $\mathcal N^{\{1\}}(R)$ represents a different state. The set $A$ can also contain more than a single element. For example:
\begin{equation}
\mathcal N^{\{1,3\}}(R) = \begin{pmatrix} 1&0&1\\1&0&0\\1&1&1\\1&1&0 \end{pmatrix},
\qquad
\mathcal N^{\{1,2,3\}}(R) = \begin{pmatrix} 1&1&1\\1&1&0\\1&0&1\\1&0&0 \end{pmatrix},
\end{equation}
and so on.}
\end{leftbar}
 
\begin{definition}{\em (Sylvester matrices)} \label{def:Sylvester}
Let $U^S$ be an $m$-dimensional unitary matrix of the form $U^S(m) \equiv H(m)/\sqrt{m}$, with $H(m)$ defined recursively as
\begin{equation}
H(2^p) := \left( 
\begin{array}{ccc}
H(2^{p-1}) & H(2^{p-1}) \\
H(2^{p-1}) & -H(2^{p-1})
\end{array} \right),
\end{equation}
for each positive integer $p$, and with $H(1):=1$.
We refer to $U^S(m)$ as \textit{normalized Sylvester matrix} and to $H(m)$ as \textit{Sylvester matrix}.
The $m$ dependence of $U^S$ and $H$ will be omitted when clear from the context.
\end{definition}

An analytic expression for the $(i,j)$ element of a Sylvester matrix can be given in the form:
\begin{equation}
[H(2^p)]_{i,j} = (-1)^{i_B \odot j_B},
\label{eq:sylvesterMatrixElements}
\end{equation}
where $i_B$ and $j_B$ are the binary representations of $i-1$ and $j-1$ respectively,
and $\odot$ is the bitwise dot product.

\subsection{Suppression law for Sylvester matrices}

In this Section we present a test which generalizes that of \cite{Crespi15}, predicting a higher fraction of suppressed pairs. This test can be assessed with a computational cost increasing only polynomially in $m$ and $n$. A similar test was proposed in \cite{Dittel2017}, using a different formalism.

Let us begin with the following two straightforward Lemmas, the proof of which we leave as an exercise for the reader.

\begin{lemma} \label{lem:perms}
Let $S_n$ be the set of permutations of $\{1,\dots,n\}$, and let $\tau \in S_n$ be a permutation different from the identity such that $\tau^2 = \mathbb 1$. Then we can uniquely associate to each $\sigma \in S_n$ another (different) permutation $\sigma_\tau \equiv\tau \circ \sigma$, where $\circ$ denotes the composition of permutations. 	
\end{lemma}

\begin{lemma} \label{lem:propsofN}
Let $\tilde r$, $R$ and $\mathcal N^A$ be as in Definitions \ref{def:Reps} and \ref{def:negation}. Then $\mathcal N^A(\tilde r) = \tilde r$ if and only if $\mathcal N^A(R) \sim R$. This, in turn, happens if and only if there is a permutation $\tau \in S_n$ such that $\mathcal N^A(R) = R^{\tau}$, where $R^{\tau}$ is obtained from $R$ by applying the permutation $\tau$ to the rows of $R$. Now, if $\tau$ is such a permutation, we have
\begin{enumerate}
\item $\tau \neq \mathbb 1$, 
\item $\tau^2=\mathbb 1$, 
\item for each $\sigma \in S_n$, $\mathcal N^A (R^\sigma) = R^{\tau \circ \sigma}$
\item all columns of $R$ in $A$ have an equal number of 1s and 0s,
\item all columns of $R$ not in $A$ have an even number of 1s and an even number of 0s.
\end{enumerate}
\end{lemma}

We are now ready to state the main result of the Section:

\begin{theorem} \label{thm:main}
Let $\ket{r}$ and $\ket{s}$ be two states of $n$ particles in $m=2^p$ modes, with corresponding MAL representations $\tilde r$ and $\tilde s$, and BM representations $R$ and $S$. If there is a subset $A$ of the columns of $R$ such that
\begin{subequations}
\label{eq:suppressionCondition}
\begin{empheq}[left=\empheqlbrace]{align}
& \mathcal N^A(R) \sim R, \label{eq:suppressionConditionForR}\\
& \displaystyle\bigoplus_{k=1}^n \bigoplus_{\alpha \in A} S_{k\alpha} = 1
\label{eq:suppressionConditionForS},
\end{empheq}
\end{subequations}
then the transition from $\ket{r}$ to $\ket{s}$ (and consequently also that from $\ket{s}$ to $\ket{r}$), when transversing the Sylvester interferometer, is suppressed.
\end{theorem} 

\begin{proof}
The transition amplitude from $\ket{r}$ to $\ket{s}$ can be written as
\begin{equation}
\mathcal A := 
\frac{1}{\sqrt{r_1!\dots r_m! s_1!\dots s_m!}}
\operatorname{Per} \left(U_{r, s}\right),
\label{eq:scatteringAmplitudeWithPermanents}
\end{equation}
where $\operatorname{Per}(U)$ denotes the permanent of $U$. In the equation above, $U_{r, s}$ is the matrix defined element-wise as
\begin{equation}
\left[ U_{r, s} \right]_{i,j} := \frac{1}{\sqrt{m}}
\left[H (2^p)\right]_{\tilde r_i, \tilde s_j},
\label{eq:matrixInPermanentFormula}
\end{equation}
where, recall, $H(2^p)$ is the Sylvester matrix (cf.\ Definition \ref{def:Sylvester}). 

Using Eqs.\ (\ref{eq:sylvesterMatrixElements}), (\ref{eq:scatteringAmplitudeWithPermanents}), and (\ref{eq:matrixInPermanentFormula}), and denoting $\textbf M_i$ the $i$th row of matrix $M$, we obtain{}
\begin{equation}
\mathcal{A} =
D \sum_{\sigma \in S_n} \prod_{k=1}^n
(-1)^{\textbf R_{\sigma(k)} \odot \textbf S_k}
= D \sum_{\sigma \in S_n} (-1)^{\mathcal{E}(\sigma)},
\label{eq:sylvesterPermanentExpansion}
\end{equation}
where $\textbf B \odot \textbf C$ denotes the \textit{bitwise dot product} between vectors $\textbf B$ and $\textbf C$, defined as $\textbf B \odot \textbf C := \bigoplus_{\alpha=1}^p B_{\alpha} C_{\alpha}$, $D$ is a constant factor, and we defined
\begin{equation}
\mathcal{E}(\sigma)
\equiv \bigoplus_{k=1}^n \textbf R_{\sigma(k)}\odot \textbf  S_k
= \bigoplus_{k=1}^n \bigoplus_{\alpha=1}^p R_{\sigma(k),\alpha} S_{k,\alpha}.
\label{eq:sylvesterPermanentExponent}
\end{equation}
The actual value of $\mathcal{E}(\sigma)$ is unimportant here, only its parity matters. Since we are interested in whether $\mathcal A = 0$, we will ignore the constant factor $D$ in Eq. (\ref{eq:sylvesterPermanentExpansion}). 

Clearly, for $\mathcal A$ to vanish, we need exactly half of the permutations to be such that $(-1)^{\mathcal E(\sigma)}=1$. A necessary and sufficient condition for this to hold is if, for each permutation $\sigma$, we can uniquely assign another permutation $\sigma'$ such that $\mathcal E(\sigma') = 1 \oplus \mathcal E(\sigma)$. From Lemmas \ref{lem:perms} and \ref{lem:propsofN}, we know that if condition (\ref{eq:suppressionConditionForR}) holds we can uniquely associate to each $\sigma$ another permutation $\sigma_\tau \equiv \tau \circ \sigma$, where $\tau$ is the permutation such that $\mathcal N^A(R)=R^\tau$. Using $\sigma_\tau$ in Eq.\ (\ref{eq:sylvesterPermanentExponent}) we have
\begin{equation}
\mathcal E(\sigma_\tau) =
\bigoplus_{k=1}^n \bigoplus_{\alpha=1}^p
R_{\tau(\sigma(k)),\alpha} S_{k,\alpha} =
\bigoplus_{k=1}^n \left[
\left( \bigoplus_{\alpha \in A} R_{\tau(\sigma(k)),\alpha} S_{k,\alpha} \right)
\oplus
\left( \bigoplus_{\alpha \notin A} R_{\tau(\sigma(k)),\alpha} S_{k,\alpha} \right)
\right].
\label{eq:sylvesterEsigmatau}
\end{equation}
Using now the explicit expression for $\mathcal N^A(R)$ and Lemma \ref{lem:propsofN}, we have
\begin{equation}
\mathcal N^A(R^\sigma)=R^{\tau\circ\sigma}
\Longleftrightarrow
\begin{cases}
1\oplus R_{\sigma(k),\alpha} = R_{\tau(\sigma(k)),\alpha}, & \alpha \in A,\\
R_{\sigma(k),\alpha} = R_{\tau(\sigma(k)),\alpha}, & \alpha \notin A.
\end{cases}
\end{equation}
Inserting these into Eq.\ (\ref{eq:sylvesterEsigmatau}), we obtain
\begin{equation}
\mathcal E(\sigma_\tau) =
\mathcal E(\sigma) \oplus
\left[ \bigoplus_{k=1}^n \bigoplus_{\alpha \in A} S_{k,\alpha} \right] = \mathcal E(\sigma) \oplus 1,
\end{equation}
where in the last step we used Eq.\ (\ref{eq:suppressionConditionForS}). Using this last result into Eq. (\ref{eq:sylvesterPermanentExpansion}) we conclude that
\begin{equation}
\mathcal A = C \sum_{\sigma \in S_n} (-1)^{\mathcal E_{R,S}(\sigma)}
= C \,\,\, \sum_{\mathclap{\sigma \in S_n : \, \mathcal E(\sigma) \text{ even}}} \,\,\,
\left[ (-1)^{\mathcal E(\sigma)} + (-1)^{\mathcal E(\sigma_\tau)} \right]
= 0,
\end{equation}
which proves that the input-output pair $(\ket{r}, \ket{s})$ is suppressed. 
\end{proof}

\begin{leftbar}
\textbf{Example 4}. 
\textit{Consider the state $\ket{r}=\ket{1,1,1,1,0,0,0,0}$ from Example 2. We have $\mathcal N^{A} (\tilde{r})= \tilde r$ for $A=\{2\}, A=\{3\}$. and $A=\{2,3\}$. Theorem \ref{thm:main} predicts suppression of all output states whose BM representation has an odd number of 1s in either the second column, the third column, or in the second and third columns combined. For example, the states $ \tilde s =(3,6,7,8)$, $\tilde s =(2,6,7,8)$, and $\tilde s =(4,6,7,8)$, having BM representations
\begin{equation}
\begin{pmatrix} 0&1&0\\1&0&1\\1&1&0\\1&1&1 \end{pmatrix},
\qquad
\begin{pmatrix} 0&0&1\\1&0&1\\1&1&0\\1&1&1 \end{pmatrix},
\qquad \textrm{and} \qquad
\begin{pmatrix} 0&1&1\\1&0&1\\1&1&0\\1&1&1 \end{pmatrix},
\end{equation}
respectively, are all suppressed.}
\end{leftbar}

\begin{remark} (Efficiency)
To check if Theorem \ref{thm:main} applies to a given input-output pair, one has to verify condition \ref{eq:suppressionCondition} for each one of the $2^p-1=m-1$ possible (non-empty) subsets of the $p$ columns of $R$ and $S$, which requires only a polynomial (in $n$ and $m$) number of elementary operations. Hence, the proposed suppression law is efficiently verifiable. 
\end{remark}

While Eq.\ (\ref{eq:suppressionConditionForS}) gives a sufficient condition for an input-output pair to be suppressed, it is not necessary. For most input states, not \textit{all} suppressed outputs satisfy Theorem \ref{thm:main}. In the next Section we give estimates of the which fraction of all states our test identifies as suppressed.

\subsection{Estimates on fraction of suppressed states}
\label{sec:suppressedfraction}

In this Section, we give estimates on the fraction of suppressed input-output pairs identified by the conditions of Theorem \ref{thm:main}, particularly focusing on upper bounds and asymptotic limits (in the number of modes $m$ and photons $n$). Throughout this Section, we are still restricted to $m=2^p$ for some integer $p$, which means all binary matrices are $n\!\times\! p$, and to $n \leq m$. The fractions we obtain consider mainly the set of all possible states, $G_{n,m}$, which has $|G_{n,m}|=\binom{m+n-1}{n}$ elements. At the end of this Section we discuss the applicability of our results to the restriction of no-collision states. 

We begin by restating the two conditions of Theorem \ref{thm:main} informally, as they would be used in a test. For simplicity, we refer to the states $\ket{s}$ and $\ket{t}$ of the Theorem as input and output states, respectively, although any conclusions can be extended to the case where the roles of input and output are reversed. With this in mind, we note that Condition \ref{eq:suppressionConditionForR} concerns only inputs, and we restate it as:
\newline

\textbf {Condition I}: For a given input with BM representation $R$, check whether $R$ has any subsets of columns $A$ such that negating those columns of $R$ results in $R$ up to a permutation of the rows (i.e.\ $\mathcal N^A(R) \sim R$).
\newline

Let us call $\mathcal A$ the set of all such subsets of columns $A$. For the state in Example 2, for instance, $\mathcal A = \left \{ \{2\}, \{3\}, \{2,3\} \right\}$. If Condition I finds no such $A$, the test fails to identify any suppressed transitions. Also, as a consequence of Lemma \ref{lem:propsofN}, the test only works if $n$ is even. Assuming that Condition I yielded some non-empty $\mathcal A$, we can test for Condition \ref{eq:suppressionConditionForS}, which is a test only on the outputs and which we restate as:
\newline

\textbf {Condition II}: Consider any output with BM representation $S$ and any $A \in \mathcal A$. If the columns $A$ of $S$ contain an odd number of 1s, the transition from $R$ to $S$ is suppressed.
\newline

For simplicity, we begin by estimating how many output states are suppressed {\em given that} Condition I identified a non-empty set $\mathcal A$ for some input. Before that, we need one final definition. We say that the elements of $\mathcal A$ are {\emph independent} if none can be replaced by a sequence of the others. That is, if there is no $A \in \mathcal A$ and $A_1, A_2, \ldots, A_k \in \mathcal A \backslash \{A\}$ such that $\mathcal N^{A_1}(\mathcal N^{A_2}  \ldots (\mathcal N^{A_k}(X))) = \mathcal N^{A} (X)$ for all binary matrices $X$. To illustrate this, consider Example 2. There, $\mathcal A = \left \{ \{2\}, \{3\}, \{2,3\} \right\}$. Clearly, negating columns $\{2,3\}$ of a binary matrix is the same as negating column $\{2\}$ followed by column $\{3\}$. Furthermore, Condition II can only be satisfied for $\{2,3\}$ if it is satisfied by either $\{2\}$ or $\{3\}$. Thus, including $\{2,3\}$ in $\mathcal A$ does not give any new suppressions beyond those identified by $\{2\}$ and $\{3\}$, so we can safely drop it (we could have dropped either $\{2\}$ or $\{3\}$ instead, to the same effect). Since the binary matrices have $p$ columns, there can be at most $p$ independent elements in $\mathcal A$. We are now ready to state the following.

\begin{corollary} \label{cor:outputfraction}
	Let $\mathcal A$ be the set identified by Condition I for some input state, and suppose it contains $q$ independent elements. Then the fraction of outputs in $G_{n,m}$ (i.e.\ including collision states) that Condition II identifies as suppressed is equal to $1 - \frac{1}{2^q} + \textrm{O}\left(\frac{\log{m}}{n}\right)$.
\end{corollary}

\begin{proof}
Suppose initially that there is a single element $A \in \mathcal A$, say $A=\{1\}$. The corresponding suppressed outputs are those whose BM representation contains an odd number of 1s in the first column. These consist of approximately half of all possible states, which can be seen as follows.
Let us write the binary matrix $S$ of some such output as
\begin{equation}
S = \left( \begin{array}{cc}
\mathbf{S_1} & S' \end{array} \right),
\end{equation}
where $\mathbf{S_1}$ is its first column and $S'$ a matrix of the remaining $p-1$ columns. Since all matrices that are equivalent up to a permutation of the rows correspond to the same state, we can assume without loss of generality that the 1s in $\mathbf{S_1}$ occupy the first slots. This means there are only $n+1$ possibilities for $\mathbf{S_1}$, and it is easy to see that $n/2$ of them satisfy Condition II. Since this holds irrespective of the choice of $S'$, we conclude that $(n/2)/(n+1) = 1/2 + \textrm{O}(1/n)$ out of {\em all} states are suppressed. The argument follows through almost unchanged for any $A$, even if it spans several columns.

Suppose now $\mathcal A$ has $q$ independent elements. By the previous paragraph, the first element of $\mathcal A$, let us call it $A_1$, leads to a suppression of approximately (i.e.\ up to O$(1/n)$) half of all outputs. The second element, $A_2$, also leads to a suppression of approximately half of all outputs---but now there is an overlap with those identified by $A_1$. Since the two are independent, approximately half of the elements identified by $A_2$ have already been identified by $A_1$ (e.g., approximately half of the matrices with an odd number of 1s in the first column also have an odd number of 1s in the second column). Thus the new suppressions identified by $A_2$ correspond only to $1/4 + O(1/n)$ of all states. Each subsequent independent element of $\mathcal A$ further divides the remaining set of unsuppressed states by half, and we conclude that Condition II in fact identifies $1 - 1/2^q + O(q/n)$ of all states as suppressed. Since $q\leq p = \log(m)$, this gives the claimed asymptotic behavior.
\end{proof}

Corollary \ref{cor:outputfraction} shows that any input which satisfies Condition I has at least 1/2 of all outputs suppressed, and this fraction can, in principle, be as high as $1-1/m$ if Condition I identifies $p$ independent elements in $\mathcal A$. It is thus essential to identify how many inputs effectively satisfy Condition I in order to determine the overall fraction of suppressed pairs. Indeed, Theorem \ref{thm:main} treats input and outputs asymmetrically, and as a consequence Condition I is much more stringent than Condition II.

\begin{corollary} \label{cor:inputfraction}
	Only an exponentially-vanishing subset of the inputs is detected by the test of Theorem \ref{thm:main}.
\end{corollary}
 
\begin{proof}
Let us begin by counting how many inputs have a particular element $A$ in their set $\mathcal A$. Consider initially that $A = \{1\}$. For Condition I to hold in this case, we need half of the elements of the first column of $R$ to be 1 (cf.\ Lemma \ref{lem:propsofN}). Since we are free to rearrange the rows as desired, we can assume that $R$ is ordered as follows
\begin{equation}
R = \left( \begin{array}{cc}
\mathbf{1_{n/2}} & R_1\\
\mathbf{0_{n/2}} & R_2\\ 
\end{array} \right),
\end{equation}
where $\mathbf{1_{n/2}}$ and $\mathbf{0_{n/2}}$ are vectors of $n/2$ 1s and 0s respectively, and $R_1$ and $R_2$ are $(n/2) \times (p-1)$ binary matrices. For $R$ to satisfy $\mathcal N^{\{1\}}(R) \sim R$, we must have $R_1 \sim R_2$, and in fact we can further reorder the rows of $R$ have such that $R_1 = R_2$. We can now count how many binary matrices satisfy these constraints. Clearly we have no choice over the first column and over $R_2$, so we only need count all possibilities for $R_1$. Given the form chosen for $R$ above, it is clear that any two choices for $R_1$ that are equal upon permutation of the rows represent the same state, and should not be counted twice. This leads to a cumbersome combinatorial problem, since we need to tally all possibilities for $R_1$ up to permutations, but the number of permutations changes depending on whether $R_1$ has repeated rows. A shortcut to this calculation is to realize that we can formally map $R_1$ back to a MAL representation, as in Section \ref{sec:notation} and, subsequently to a quantum state of $n/2$ photons in $2^{p-1}=m/2$ modes. Thus, the number of possibilities for $R_1$ is $\binom{(n+m)/2 - 1}{n/2}$.

Consider next another possibility, that $A = \{1,2\}$. The counting in this case is similar to before, but a little trickier, because we need to keep track of how many choices we have for the first two columns such that $\mathcal N^{\{1,2\}}(R) \sim R$. Let us once more order the matrix as follows:
\begin{equation}
R = \left( \begin{array}{ccc}
\mathbf{1_{n/2}} & \mathbf{a_1} & R_1\\
\mathbf{0_{n/2}} & \mathbf{a_2} & R_2\\ 
\end{array} \right),
\end{equation}
where $\mathbf{1_{n/2}}$ and $\mathbf{0_{n/2}}$ are defined as before, $\mathbf{a_1}$ and $\mathbf{a_2}$ are two binary vectors of length $n/2$ and $R_1$ and $R_2$ are $(n/2) \times (p-2)$ binary matrices. We want $R$ to be ordered in such a way that $\mathcal N^{\{1,2\}}(R) \sim R$ implies $R_1 = R_2$, to reuse part of the previous argument. This immediately implies that $\mathbf{a_2}$ is obtained from $\mathbf{a_1}$ by negating its elements. We know that $\mathbf{a_1}$ and $\mathbf{a_2}$, together, must contain $n/2$ 1s, due to Lemma \ref{lem:propsofN}, but that is already guaranteed by the fact that they are negations of each other\footnote{Note that, whenever $\mathbf{a_1}$ has 1/4 of the 1s, we are double-counting some matrices that were already included in the case $A = \{1\}$. But in the asymptotic limit they form a negligible fraction of the cases, so we do not worry about this correction.}. From the freedom of reordering rows, we can assume that the 1s in $\mathbf{a_1}$ occupy the first positions. This leaves $(n/2+1)$ possibilities for $\mathbf{a_1}$. Combining that with the $\binom{(2n+m)/4 - 1}{n/2}$ possibilities for $R_1$ as in the previous paragraph, we get a total of $(n/2+1) \binom{(2n+m)/4 - 1}{n/2}$ possibilities. How does this compare with the case $A=\{1\}$? We can use the following asymptotic expression for the binomial coefficient $\binom{n}{k}$ 
\begin{equation}
\binom{n}{k} \approx \frac{(n-k/2)^k}{k^k e^{-k}\sqrt{2 \pi k}},
\end{equation}
which holds when $n$ is both large and much larger than $k$. By using this expression it is easy to see that, in the limit of both $m$ and $n$ large with $n \leq m$, $(n/2+1) \binom{(2n+m)/4 - 1}{n/2}$ grows exponentially slower than  $\binom{(n+m)/2 - 1}{n/2}$, and so we can use the latter as an upper bound. In fact, as we consider sets $A$ comprising of more columns, the constraints tend to become more restrictive and the number of matrices that satisfy them decreases. So we will use $\binom{(n+m)/2 - 1}{n/2}$ as an upper bound on the number of states that satisfy Condition I for {\em any} $A$.

We are now ready to give an upper bound on the number of states that satisfy Condition I for some nonempty $\mathcal A$. Since there are $p$ columns in the binary matrices, There are $2^p-1 = m-1$ possible $A$s that can appear in $\mathcal A$. By the inclusion-exclusion principle, we would need to sum the number of states that satisfy Condition I for each possible $A$, then subtract those that have been counted multiple times because they satisfy it for more than one $A$. It is simpler, however, to use the union bound, which in this case says that the number of states is upper-bounded by $(m-1)\binom{(n+m)/2 - 1}{n/2}$. Recall now that the total number of states is $\binom{n+m - 1}{n}$. Using an asymptotic formula for the binomial coefficient, it is clear that the fraction of states detected by the test is exponentially small in the limit of large $n$ and $m$, as claimed.
\end{proof}

In Ref. \cite{Dittel2017}, the authors seem to reach a similar conclusion, using a different formalism, for the fraction of suppressed outputs given a specific input (i.e., Corollary \ref{cor:outputfraction}). However, they do not provide an estimate for the fraction of inputs that satisfy Condition I (i.e., Corollary \ref{cor:inputfraction}). 

So far, we have considered only the full set of states (i.e, including collision states) in the estimates of suppressed fractions, but the restriction to no-collision states is often more useful. For example, no-collision outputs are the only detected outcomes when the experiment is performed using bucket detectors (i.e.\ that do not distinguish one photon from many). More importantly, experimental implementations typically consider inputs with no more than a single photon per mode. This is also relevant for boson sampling applications, being input state with at most one photon per mode the appropriate choice for its computational hardness. Thus, it would be interesting to obtain versions of Corollaries \ref{cor:outputfraction} and \ref{cor:inputfraction} where both the suppressed pairs and the set of all states included this restriction.  Unfortunately, some pathological instances arise when we try to specialize the previous results in that way. To see that, consider the case where $n=m$. There, we have a single no-collision state, and it satisfies Condition I. Thus, we conclude that $100\%$ of inputs in that case have suppressed outputs! As we now argue, it is still possible to show a weaker version of Corollary \ref{cor:inputfraction} for no-collision inputs.

Consider the regime where $m = \textrm{O} \left(n^2\right)$. Experiments are often done in this limit, especially since it seems to be a requirement for the computational hardness of the boson sampling model \cite{Aaronson10}. It is easy to show that the set of no-collision states is not a negligibly subset of all states in this regime, due to the so-called birthday paradox. To illustrate this suppose $m=n^2$ holds exactly, in which case the fraction of no-collision states among all states is $\binom{n^2}{n}/\binom{n^2+n-1}{n}$. Using Stirling's approximation, one obtains that this tends to $1/e$ in the limit of large $n$. Since the set of no-collision states is only polynomially small in the set of all states, a no-collision version of Corollary \ref{cor:inputfraction} must still hold---even if {\it all} inputs that satisfy Condition I were concentrated in the no-collision subset, they would still be an exponentially small fraction of it. This argument shows that the conclusion of Corollary \ref{cor:inputfraction} can be extended to the no-collision case in the limit $m = \textrm{O} \left( n^2 \right)$, and we leave it as an open question whether it holds in general. 

Corollary \ref{cor:inputfraction} also has consequences for the application of Theorem \ref{thm:main} as a test for validating boson sampling experiments. As argued in \cite{Tichy14}, suppressed events in Hadamard matrices (such as the Sylvester or Fourier matrices) could be useful as a way to witness partial photonic indistinguishability. Informally, the idea is that we only have suppressions of certain transitions if the particles are perfectly identical, and so observations of quantumly-suppressed events could be used to estimate the degree of partial distinguishability of the photons. 
Corollary \ref{cor:inputfraction} shows that in Scattershot BosonSampling experiments \cite{AaronsonBlog,Lund14,Bentivegna15}, where inputs are chosen uniformly at random from all no-collision states, the number of suppressed events detected by Theorem 1 vanishes exponentially. On the other side, when specific input states that satisfy Condition I for many different $A$s are employed, Theorem 1 might provide a favorable scaling.
Indeed, we showed that as many as $1-\textrm{O}(1/m)$ out of all outcomes can be suppressed, but we leave a formal description of such a test for future work, as well as the question of whether the Sylvester matrix is optimal for this task.

\subsection{Comparison between Sylvester and Fourier matrices}
\label{sec:SylFoucomparison}

In this Section, we briefly compare the fraction of suppressed input-output pairs for Sylvester and Fourier matrices for a few small sizes (restricting ourselves to no-collision states, as these are experimentally more relevant as discussed previously). Input-output suppressions have been studied for Fourier matrices, for instance, in \cite{Tichy10}. We are not aware of any estimates of the sort we made in Section \ref{sec:suppressedfraction} for Fourier matrices, but from \cite{Tichy10} it seems that the restriction over inputs is more stringent than the one in our Condition I. This suggests that a result similar to Corollary \ref{cor:inputfraction} should holds, also limiting the use of Fourier matrices to witness photon distinguishability for arbitrary input states, as discussed in the previous Section. In Supplementary Figure \ref{fig:fourierVSsylvesterManyBosonMatrices} we see how 8-mode Sylvester and Fourier matrices compare in terms of suppressed transitions for 2 and 4 photons, where it is clear that the Sylvester matrix outperforms the Fourier one.

For small values of $n$ and $m$ it is also possible to exactly compute the fraction of suppressed input-output pairs (including those not detected by Theorem \ref{thm:main}), which we report in the Supplementary Table \ref{table:suppressionFractions}. This Table shows that the Sylvester matrix indeed seems to perform better than the Fourier matrix for most cases, although there are exceptions. For $n=3$ and $7$, for example, it is possible to show that there can be no suppressed transitions for Sylvester matrices (this is a consequence of the fact that the permanent of a $(2^p-1) \times (2^p - 1)$ matrix where all elements are $\pm 1$, for integer $p$, is never zero \cite{Wanless2005}), whereas the Fourier matrix does contain a few suppressions. Note also that, for all cases reported in Supplementary Table \ref{table:suppressionFractions} where the test of Theorem \ref{thm:main} can be applied, the {\em total} number of suppressed pairs for Fourier matrices is not only smaller than the corresponding quantity for Sylvester matrices, but smaller even than the subset of input-output pairs that our test detects. 

\begin{suppFig}[ht!]
\centering
\hspace*{-1cm}
\includegraphics[width=0.8\textwidth]{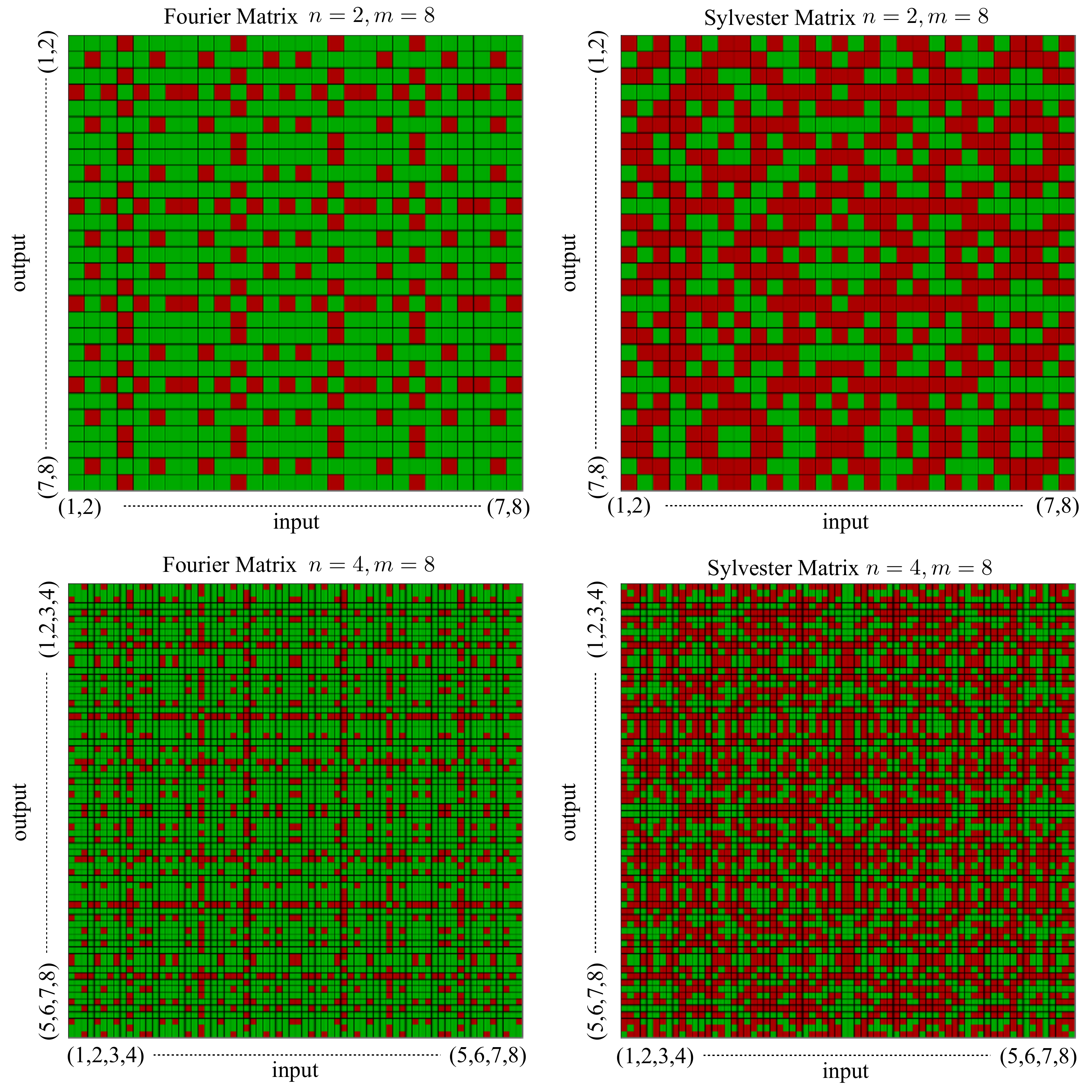}
\caption{Matrices of scattering amplitudes between all combinations of input-output pairs of no-collision states of 2 and 4 photons into 8 modes, for Sylvester and Fourier matrices. Red and green squares represent suppressed and non-suppressed pairs respectively.}
\label{fig:fourierVSsylvesterManyBosonMatrices}
\end{suppFig}

{\renewcommand{\arraystretch}{1.2}
\begin{suppTab}
\begin{center}
  \begin{tabular}{  c  c  c  c  c }
    \hline
    \; $m$ \; & \; $n$ \;  & Sylvester (All) & Sylvester (Test) & \; Fourier\; \\
    \hline
    \multirow{2}{*}{4}& 2 & 66.7 \% & 66.7 \% & 44.4 \% \\
     & 3 & 0 & 0 & 0  \\ \hline
    \multirow{6}{*}{8}& 2 & 57.14\% & 57.14\% & 24.49\%  \\
    & 3 & 0 & 0 & 16.33\% \\
    & 4 & 54.86\% & 27.42\% & 21.22\% \\ 
    & 5 & 57.14\% & 0 & 24.49\%  \\ 
    & 6 & 57.14\%  & 57.14\% & 48.98\%  \\
    & 7 & 0 & 0 & 0  \\ \hline
    \multirow{6}{*}{16}& 2 & 53.33\% & 53.33\% & 14.22\% \\
    & 3 & 0 & 0 &  5.22\%  \\
    & 4 & 40.57\% & 11.5\% & 5.37\%  \\ 
    & 5 & 40.57\% & 0 & 2.54\% \\ 
    & 6 & 37.14\% &  $9.9\%^*$ &2.32\% \\
    & 7 & 0 & 0 & $0.88\%^*$   \\ 
    & 8 & 26.24\% & $6.9\%^*$ & 1.20\%   \\ 
  \end{tabular}
  \caption{Fractions of suppressed pairs for Sylvester and Fourier matrices, for several values of $m$ modes and $n$ photons. The column Sylvester (All) reports all suppressed pairs, and the column Sylvester (Test) only those detected by Theorem \ref{thm:main}. Note that the test described in the Theorem only works for even values of $n$.
  Numbers indicated by an ${}^*$ are estimates obtained by sampling 500000 different input-output pairs, all others are exact.}
\label{table:suppressionFractions}
\end{center} 
\end{suppTab}}

\section{Photon generation, manipulation and detection}
\label{sec:Photondetails}

Single photons were generated at 785 nm with a type-II parametric down-conversion process in four PDC sources
for scattershot configuration, pumping two crystals (2-mm long BBO) with a 392.5 nm wavelength Ti:Sa pulsed
laser. Photons are spectrally filtered by means of 3 nm interferential filters and coupled into single-mode fibers. The
indistinguishability of the photons is reached by means of a polarization compensation stage and by propagation
through delay lines for each path before injection into the interferometer via a single-mode fiber array. After the
evolution through the integrated devices, photons are collected via a multimode fiber array. The detection system for
the scattershot experiment consists of four single-photon avalanche photodiodes for the 4-mode chip and other four for
the heralding photons in the scattershot regime. Single-shot measurements have been performed with a 2-photon state
produced by a single BBO crystal and injected in the 8-mode Sylvester interferometer. At the detection stage, eight
avalanche photodiodes have been used to collect all output combinations. An 8-channel electronic data acquisition
system (ID-800 by IDQuantique) allowed us to detect 2-photon coincidences between all output pairs and 4-photon
coincidences (two injected plus two triggered) for all possible input states. LabView and C programs have been used
to retrieve the coincidence events associated to all possible output combinations.

\section{Model of the experiment}

Here we discuss a theoretical model to describe the results of the experiment with the 4-mode device. 
In addition to the non-perfect unitary transformation, two sources of deviation from the ideal behavior
contribute to the output measured pattern: i) partial indistinguishability of the heralded photons, and ii) multi-photon
emission from the sources.

As a preliminary step we characterized the parameters of the experimental setup. Typical singles count rates for the different sources are in the range $100-250$ kHz, while two-fold coincidences are in the range $10-35$ kHz. Additionally, the overall transmission from the delay lines to the output fiber-array are $\sim 0.08-0.16$ depending on the input-output combination. From these values we estimated the nonlinear gain $g$ of the sources ($g \sim 0.12$ for C1, $g \sim 0.115$ for C2), the heralding probabilities $\eta^{T}_{i}$ (in the range $\sim 0.1-0.22$ for the different sources), and the overall transmission of the injected photon from the generation to the detection stage ($\sim 0.01-0.02$, including detection efficiency).

\subsection{Multi-photon emission}

Multi-photon emission arises due to the probabilistic nature of parametric down-conversion (PDC). Indeed, there is a non-zero probability that two pairs are emitted from the same source within the same pulse. Ignoring terms with the emission of three or more pairs, the output state of each source can be approximated as:
\begin{equation}
\vert \psi \rangle \sim \vert 0,0 \rangle + g \vert 1,1 \rangle + g^{2} \vert 2,2 \rangle,
\end{equation}
where $g$ is the nonlinear gain of the source. Note that in our experiment, each PDC crystal corresponds to two different photon-pair sources as shown in Fig. 2 of the main text.

Let us consider the situation where an event is recorded by the heralding detectors corresponding to inputs $(i,j)$, in coincidence with an event registered by the detectors placed at output modes $(m,n)$. This event is assigned to the transition from the input combination $(i,j)$ to the output one $(m,n)$. The correct evolution is obtained when the sources on modes $(i,j)$ generate a photon pair, the corresponding heralding detectors click, and two photons are detected on output modes $(m,n)$. Multi-photon emission and non-photon number resolving detectors result in additional patterns that can excite the same set of detectors. These patterns act as noise contributions that cannot be discriminated from the correct evolution. Two different contributions have to be considered.
a) Three different sources connected to input ports $(i,j,k)$ emit a photon pair, while only the heralding detectors on modes $(i,j)$ click. Hence, three-photons are effectively produced. If only the detectors on output modes $(m,n)$ click (due to losses or the presence of more than one photon in the one output mode), this process cannot be discriminated from the correct evolution $(i,j) \rightarrow (m,n)$.
b) Only the sources connected to input ports $(i,j)$ generate photons, but one of the two sources produces a double-pair event. This event cannot be discriminated in the heralding process with non-photon number resolving detectors. In this case, two-photons may be injected in the same input mode. Similarly to case a), when only detectors on mode $(m,n)$ click, this process cannot be discriminated from the correct evolution $(i,j) \rightarrow (m,n)$.

\subsection{Partial photon indistinguishability}

Partial distinguishability between the generated heralded photon arises due to spectral correlations between photons belonging to the same pair. Indeed, the two-photon term of a parametric down-conversion source takes the following form:
\begin{equation}
\vert \psi^{(2)} \rangle = \int d \omega_{1} \int d \omega_{2} f(\omega_{1}, \omega_{2}) a^{\dag}_{1}(\omega_{1}) a^{\dag}_{2}(\omega_{2}) \vert 0,0 \rangle,
\end{equation}
where $f(\omega_{1},\omega_{2})$ is the two-photon spectral amplitude. In general, spectral correlations are encoded in the function $f(\omega_{1}, \omega_{2})$. When PDC sources are adopted as heralded single-photon sources, one of the two photons is detected to certify the presence of the twin photon. Due to the correlations encoded in the two-photon wave packet, the heralded photon will be in general in a mixed spectral state. This will result in a degree of partial distinguishability between the photons emitted by two identical sources. The effective joint density matrix describing the state of two heralded photons can be then approximated as:
\begin{equation}
\rho^{(2)} = p \vert 1,1 \rangle \langle 1,1 \vert + (1-p) \vert 1',1'' \rangle \langle 1',1'' \vert,
\end{equation}
where $\vert 1,1 \rangle$ stands for two indistinguishable photon, and $\vert 1',1'' \rangle$ stands for two distinguishable particles. Here, $p$ is an effective parameter describing the indistinguishability of the two photons. The parameter $p$ can be characterized from the visibility of an Hong-Ou-Mandel interference experiment performed with a 50:50 symmetric beam-splitter. In our case, the measured visibility between photons emitted from the two sources was $V^{(2)} = 0.724 \pm 0.008$. The parameter $p$ can be retrieved from the value of $V^{(2)}$ by taking into account multi-photon emission, leading to $p=0.758 \pm 0.008$.

\subsection{Comparison between model and experimental data - 4 mode interferometer}

\begin{suppFig}[ht!]
\centering
\includegraphics[width=0.99\textwidth]{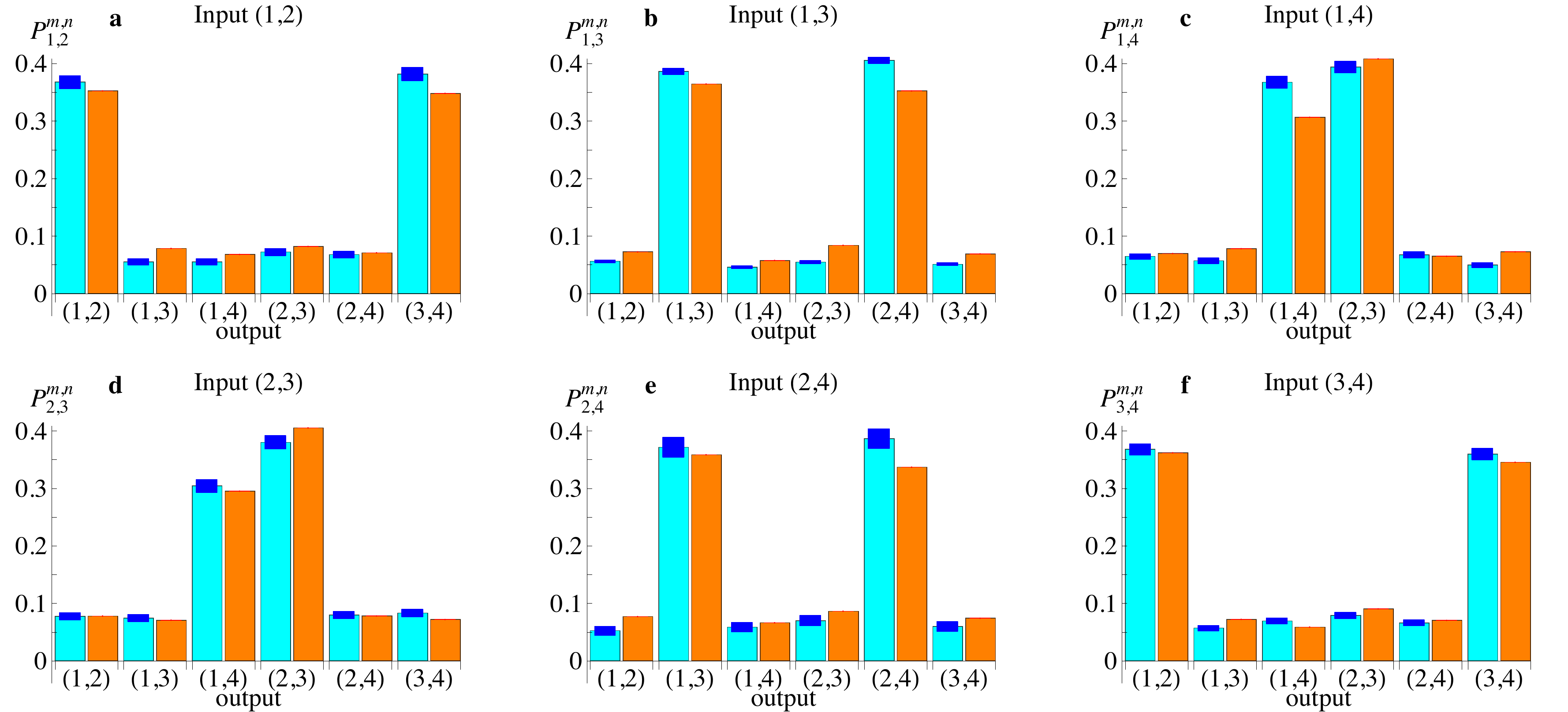}
\caption{Comparison of the experimental data (indistinguishable photons) with the theoretical model taking into account the main sources of experimental imperfections. Cyan bars: experimental data. Blue rectangles: Poissonian errors on the experimental data. Orange bars: predictions from the theoretical model.}
\label{fig:data_model_comparison}
\end{suppFig}

In Supplementary Figure \ref{fig:data_model_comparison} we report the comparison between the experimental data and the theoretical predictions obtained by the model (discussed below) which takes into account all effects i)-iii). A good agreement is found and confirmed by the variation distances $d=1/2 \sum_{\alpha} \vert P^{\mathrm{exp}}_{\alpha} - P^{\mathrm{mod}}_{\alpha} \vert$ between the experimental data $P^{\mathrm{exp}}_{\alpha}$ and the predictions $P^{\mathrm{mod}}_{\alpha}$. For the reported experiment, the value $d$ averaged over the input states is $\overline{d}=0.053 \pm 0.021$. 
Analogously, in Supplementary Figure \ref{fig:data_model_comparison_2} we report the results of the comparison between the experimental data with distinguishable photons and the predictions obtained by the model. Again the good agreement is confirmed by the distance $\overline{d} = 0.045 \pm 0.016$, averaged over the input states.

\begin{suppFig}[ht!]
\centering
\includegraphics[width=0.99\textwidth]{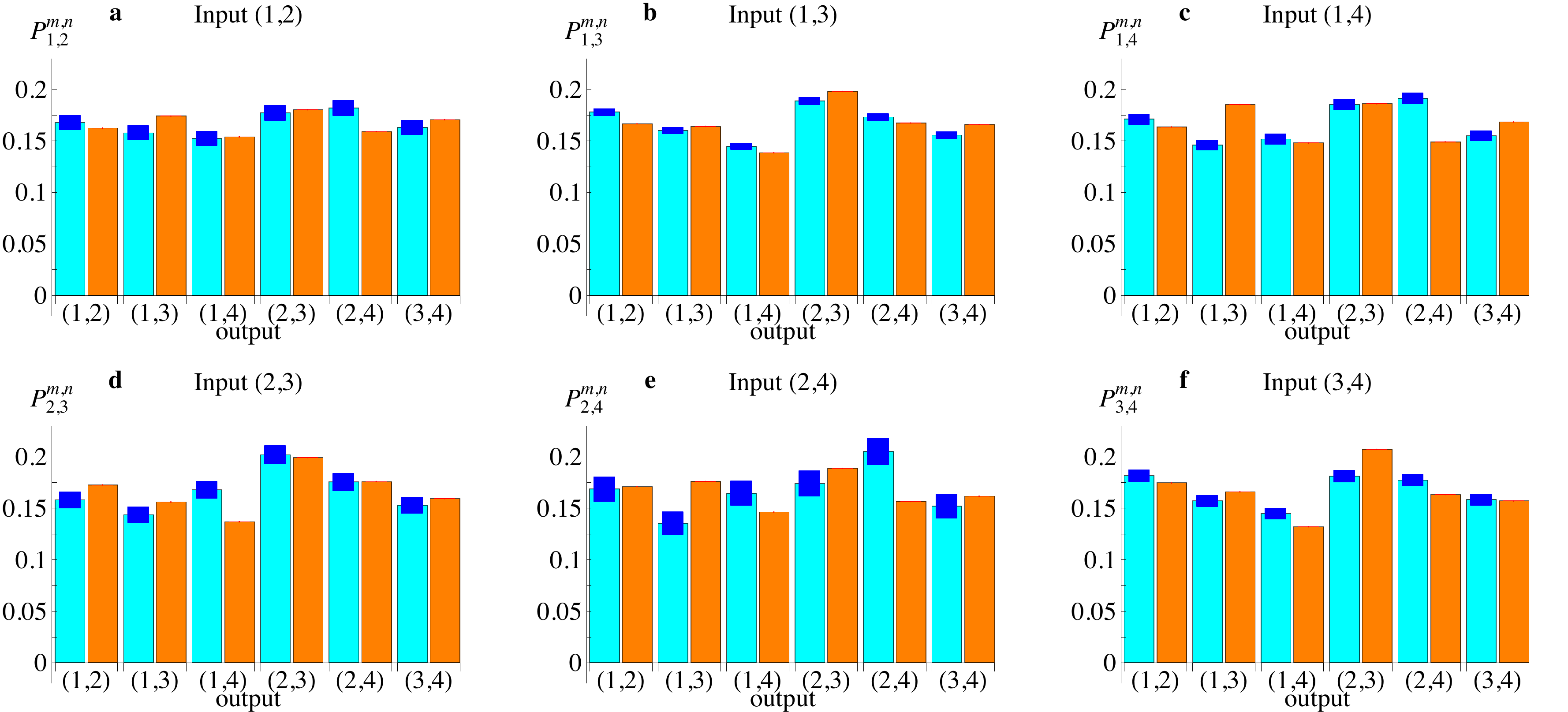}
\caption{Comparison of the experimental data (distinguishable particles) with the theoretical model taking into account the main sources of experimental imperfections. Cyan bars: experimental data. Blue rectangles: Poissonian errors on the experimental data. Orange bars: predictions from the theoretical model.}
\label{fig:data_model_comparison_2}
\end{suppFig}